%% file: paper.tex
\documentclass[11pt]{article}

\usepackage[margin=1in]{geometry}
\usepackage[T1]{fontenc}
\usepackage{lmodern}
\usepackage{microtype}
\usepackage{amsmath}
\usepackage{amssymb}
\usepackage{amsthm}
\usepackage{booktabs}
\usepackage{graphicx}

\newtheorem{definition}{Definition}
\newtheorem{proposition}{Proposition}
\usepackage{enumitem}
\usepackage{xcolor}
\usepackage[hidelinks,colorlinks=true,linkcolor=blue!50!black,
            urlcolor=blue!50!black,citecolor=blue!50!black]{hyperref}
\usepackage[capitalize,noabbrev]{cleveref}
\usepackage[numbers,sort&compress]{natbib}

\title{\bfseries Skills as Verifiable Artifacts:\\
       A Trust Schema and a Biconditional Correctness Criterion\\
       for Human-in-the-Loop Agent Runtimes}

\author{%
  Alfredo Metere\\
  Enclawed, LLC, California, USA\\
  \texttt{alfredo.metere@enclawed.com}
}
\date{\today}

\begin{document}

\maketitle

\begin{abstract}\noindent
Agent skills --- structured packages of instructions, scripts, and
references that augment a large language model (LLM) without
modifying the model itself --- have moved from convenience to
first-class deployment artifact. The runtime that loads them
inherits the same problem package managers and operating systems
have always faced: a piece of content claims a behavior; the
runtime must decide whether to believe it. We argue this paper's
central thesis up front: a skill is \emph{untrusted code} until
it is verified, and the runtime that loads it must enforce that
default rather than infer trust from a signature, a clearance, or
a registry of origin. Without skill verification, a
human-in-the-loop (HITL) gate must fire on every irreversible
call --- which is operationally untenable and degrades into
rubber-stamping at any non-trivial scale. With skill verification
treated as a separate, gated process, HITL fires only for what
is unverified, and the system becomes sustainable. We give a
trust schema (\S\ref{sec:schema}) that includes an explicit
verification level on every skill manifest; a capability gate
(\S\ref{sec:gate}) whose HITL policy is a function of that
verification level; a \emph{biconditional} correctness criterion
(\S\ref{sec:biconditional}) that any candidate verification
procedure must satisfy on an adversarial-ensemble exercise
(\S\ref{sec:eval}); and a portable runtime profile
(\S\ref{sec:guidelines}) with ten normative guidelines abstracted
from a working open-source reference implementation
\cite{metere2026enclawed}. The contribution is harness- and
model-agnostic; nothing here requires retraining, fine-tuning, or
proprietary infrastructure.
\end{abstract}


\section{Introduction}

The emergence of \texttt{SKILL.md} as a portable, harness-neutral
description of procedural knowledge for LLM agents has solved a
real problem: an agent operator can curate a library of small,
inspectable, version-controlled skill packages and ship them
across platforms. Empirically, curated skill libraries improve
agent pass-rates on benchmarks such as Mind2Web
\cite{deng2023mind2web} and Embodied Agent Interface
\cite{li2024embodied}, while indiscriminate admission of
self-generated skills can degrade or actively attack an agent's
behavior \cite{andriushchenko2025agentharm}. The asymmetry is
the same one distributed software ecosystems hit decades ago:
trust cannot be inferred from origin alone
\cite{thompson-trusting-trust}, package managers that do infer
it have been broken every way the literature predicted
\cite{cappos-look-mirror, birsan-dependency-confusion}, and the
formal response is by now textbook \cite{samuel-tuf}: typed
artifacts, signed manifests, capability sandboxing, hash-chained
audit, and least-privilege gates around side-effects.

What is \emph{not} textbook is what the runtime should do with a
skill it has just signed and clearance-checked but has \emph{not}
yet verified to behave the way its manifest claims. Most current
harnesses treat that distinction as collapsed: signature plus
clearance equals trust. We argue this is wrong, dangerously so,
and that the cost of getting it wrong is paid by the human in
the loop.

\paragraph{Skills are untrusted code.} A skill influences the
agent's prompt, tool calls, and write surfaces, making a malicious
skill a direct prompt-injection vector that persists across model
calls and survives context truncation \cite{owasp-llm-top-10,
mitre-atlas, ferrara2024genai, chen2025agentpoison}. A signature
on the manifest tells the runtime that some signer endorsed the
artifact at some point; it says nothing about whether the
artifact's behavior conforms to its declared manifest. Recent
work shows the gap is exploitable in practice: backdoors can be
smuggled into signed code-completion models in ways that pass
strong static detection \cite{yan2024codebreaker}. Treating a
signed-and-cleared skill as trusted is no more defensible than
treating a code-signed executable as malware-free because the OS
vendor signed it.

\paragraph{HITL is the universal default for unverified skills.}
Without a way to trust a skill's claims about itself, the runtime
must gate every irreversible call through human-in-the-loop
(HITL) approval. This is operationally untenable: operators under
load rubber-stamp every prompt, degrading the gate into a
fiction. The right response is \emph{not} to weaken the gate; it
is to introduce a second process --- skill verification --- whose
successful completion allows the runtime to trust portions of a
skill's manifest and relax HITL frequency. Verification happens
once at bootstrap; the result is recorded in the manifest as an
explicit \emph{verification level} and is immutable for the
session, mirroring the ``capability bootstrap discipline''
\cite{ocap-discipline} adapted to skills.

\paragraph{Side-effects must be verifiable, not just gated.}
We propose a single \emph{biconditional} criterion: the
observable side-effects of an agent run must be in 1-to-1
correspondence with the approved-and-executed set in the audit
log. Any candidate verification procedure (review, fuzz testing,
adversarial ensemble, formal analysis) must demonstrate that the
skill under test satisfies this criterion before its manifest can
be elevated above the unverified default.

\paragraph{Contributions.} This paper makes four contributions.

\begin{enumerate}
  \item A \emph{skill trust schema} (\S\ref{sec:schema}) that
        treats a skill as a tuple of manifest, content, and
        capabilities, with an explicit \emph{verification level}
        field and a no-runtime-mutation discipline.
  \item A \emph{capability-gate model} (\S\ref{sec:gate}) whose
        HITL policy is a function of the loaded skill's
        verification level: unverified skills incur HITL on every
        irreversible call; verified skills incur HITL only for
        capabilities outside the verified manifest.
  \item A \emph{biconditional correctness criterion}
        (\S\ref{sec:biconditional}) that any candidate skill-
        verification procedure must satisfy, with a
        characterization of the failure modes it catches and
        the ones it does not.
  \item A \emph{portable runtime profile} (\S\ref{sec:guidelines})
        with ten normative guidelines abstracted from a working
        open-source reference implementation
        \cite{metere2026enclawed}, including a no-bypass-switch
        invariant and an untrusted-by-default rule for skill
        admission.
\end{enumerate}

We sketch an adversarial-ensemble evaluation
(\S\ref{sec:eval}) under which the criterion is exercised --- and
under which a candidate skill-verification procedure can be
benchmarked --- and discuss open problems in
\S\ref{sec:discussion}.

The schema is deliberately model- and harness-agnostic. It assumes
nothing about training, fine-tuning, or RLHF, and assumes only the
weakest run-time interface: that an agent harness invokes tools
through a typed dispatch step the runtime can interpose on. Every
existing harness adopting \texttt{SKILL.md} satisfies this
constraint.

\paragraph{Reference implementation.} We draw both the schema and
the criterion from a working open-source framework,
\emph{enclawed} \cite{metere2026enclawed}, which hard-fork-hardens
a single-user AI assistant gateway with each of
the primitives this paper
abstracts: a Bell--LaPadula classification scheme, an Ed25519
signed-module loader with a clearance-bounded trust root, a
hash-chained audit log, an egress guard, a regex-based DLP
scanner, a HITL controller with a checkpointable agent-session
state machine, and a transaction buffer with rollback. Section
\ref{sec:guidelines} extracts a small set of normative guidelines
from that reference implementation, intended as a starting point
for a portable runtime profile any \texttt{SKILL.md}-adopting
harness could implement to bound the attack surface skills
introduce.

\section{Threat model: skills as a new supply-chain surface}
\label{sec:threat}

We assume an agent operator who already runs LLM agents with tool
access in an environment where \emph{some} of the tools have
real-world side-effects (file deletion, message sending, payment,
database writes, on-chain operations). The skill ecosystem the
operator subscribes to is partially trusted: the operator's own
curated skills are signed by an operator-controlled root, but
upstream registries and self-generated skills are not.

\paragraph{Adversaries.} We consider three:

\begin{enumerate}
  \item A \emph{registry attacker} who publishes a malicious skill
        that the operator's harness pulls and runs. The skill's
        instructions induce the agent to take harmful side-effects
        in the operator's environment.
  \item A \emph{prompt-injection attacker} whose payload reaches
        the agent through a benign skill's data dependency
        (e.g.~a fetched document the skill instructs the agent to
        summarise) and re-routes the agent's tool calls.
  \item A \emph{self-generated-skill attacker}: the agent itself,
        under model error or adversarial input, synthesises a new
        skill, registers it, and uses it. This is not a malicious
        actor in the usual sense, but operationally indistinguishable
        from one.
\end{enumerate}

\paragraph{Out of scope.} We do not assume an attacker who has
compromised the operator's signing key, the runtime binary, or the
audit-log storage out-of-band. We do not assume a defended
hardware trust anchor (a TPM, secure element, or equivalent),
though such an anchor improves every layer below. We do not
discuss model-weight-level attacks (poisoning, jailbreak from
training data), which are orthogonal.

\paragraph{Goal.} The runtime should:
\begin{itemize}
  \item refuse to load a skill whose manifest does not verify;
  \item refuse to dispatch a tool call whose required capability
        the loaded skill did not declare;
  \item refuse to execute an irreversible call without an
        explicit, audited decision;
  \item produce an audit trail under which the biconditional
        criterion of \S\ref{sec:biconditional} can be checked.
\end{itemize}

\section{The skill trust schema}
\label{sec:schema}

We define a skill artefact as a tuple
\[
  \texttt{Skill} \;=\; (\,M,\; \texttt{content},\; \sigma\,)
\]
where $M$ is a manifest, \texttt{content} is the body of the skill
(typically the \texttt{SKILL.md} file plus referenced scripts), and
$\sigma$ is a detached signature over the canonical bytes of $(M,
\texttt{content})$ produced by a signer in the operator's
trust root.

The manifest $M$ has five mandatory fields:

\begin{description}
  \item[$M.\mathrm{label}$] a classification label
        $\langle \ell, C, R \rangle$ in a Bell--LaPadula style
        lattice \cite{bell-lapadula}, with rank~$\ell$, compartment
        set~$C$, and releasability caveats~$R$. Labels combine via
        join: $a \sqcup b = \langle \max(\ell_a,\ell_b), C_a \cup
        C_b, R_a \cap R_b \rangle$.
  \item[$M.\mathrm{caps}$] a finite set of declared capabilities
        drawn from a fixed vocabulary; see \S\ref{sec:caps}.
  \item[$M.\mathrm{signer}$] a key identifier referencing an entry
        in the trust root.
  \item[$M.\mathrm{version}$] a monotone integer; replays of an
        older signed manifest are rejected if a newer one with the
        same identity has been observed.
  \item[$M.\mathrm{verification}$] the skill's verification level;
        see \S\ref{sec:verification}.
\end{description}

\subsection{Verification levels}
\label{sec:verification}

A signed manifest binds an artifact to a signer's name; it does
not bind the artifact's behavior to the manifest's claims. We
make this distinction explicit in the schema. Every skill
manifest carries a verification level $M.\mathrm{verification}$
drawn from a fixed enum of four values: \textsc{unverified},
\textsc{declared}, \textsc{tested}, and \textsc{formal}. The
default is \textsc{unverified}; higher levels grant the runtime
more latitude when policing the skill's tool calls
(\S\ref{sec:gate}).

\paragraph{\textsc{unverified}.} No claim is made about the
skill's behavior beyond its existence. Default for any skill the
operator has not explicitly examined. The runtime treats this
skill as untrusted code: every irreversible capability call goes
through the HITL gate, regardless of whether the manifest
declares the capability or not.

\paragraph{\textsc{declared}.} A trust-root signer has examined
the skill's content and \emph{attests} that the skill's
side-effects are bounded by its declared capability set
$M.\mathrm{caps}$. The attestation is the signer's professional
reputation; the runtime trusts capabilities in $M.\mathrm{caps}$
but still requires HITL on irreversible calls outside that set.

\paragraph{\textsc{tested}.} In addition to a declaration, the
skill passes an \emph{adversarial-ensemble} verification run
(\S\ref{sec:eval}) under which the biconditional criterion of
\S\ref{sec:biconditional} holds. The runtime permits irreversible
capabilities in $M.\mathrm{caps}$ without per-call HITL but logs
every such call.

\paragraph{\textsc{formal}.} A formal analysis tool has produced
a machine-checkable proof that the skill's behavior is a subset
of its declared capability set under the runtime's threat model.
This level is aspirational at the time of writing; we include it
for completeness because the schema field is fixed-width and
adding it later requires a manifest version bump.

\paragraph{Why levels, not a continuum.} A continuous
``trust score'' would allow the agent to argue itself into
permission gradients; a fixed enum forces a discrete decision
that can be audited and revoked atomically. The same reasoning
applies to the configuration profiles in
\S\ref{sec:profile}.

\paragraph{Verification is an offline, bootstrap-time concern.}
Verification levels are set by the operator (or the operator's
delegate, e.g., a CI pipeline running an adversarial-ensemble
suite over each candidate skill) \emph{before} the runtime
starts, and recorded into the signed manifest at the
\textsc{declared} level or above. The runtime never elevates a
verification level during a session; doing so would re-introduce
the trust-bootstrap problem the schema is designed to avoid. A
manifest that arrives at runtime claiming a verification level
above what its signer is authorized to attest is rejected.

\subsection{No agent mutation of skills at runtime}
\label{sec:no-mutation}

Once a skill is loaded, its content and manifest are immutable
for the lifetime of the agent session: the agent cannot edit a
skill, append to its manifest, or upgrade its verification level.
Any attempt the agent makes to modify a loaded skill --- through
a file-system tool call, a runtime API, or any indirect
side-effect on the on-disk artifact --- is itself an irreversible
capability that must walk the HITL gate \emph{and} be recorded
in the hash-chained audit log with the pre- and post-mutation
content hashes and the operator's decision, regardless of whether
the request is approved or denied. On approval, the modification
produces a \emph{new} skill artifact that must be re-verified
before the next session. Without the immutability rule, an agent
that admits one verified skill could rewrite that skill's content
in place and silently change what ``verified'' means for the
rest of the session; without the audit trail, the rewrite would
be invisible to post-incident review.

\subsection{The trust root}

The trust root is a finite, append-only set of signer entries
$\{(\mathrm{keyId}_i, \mathrm{pubKey}_i, \mathrm{maxClearance}_i)\}$.
A signer is trusted to sign skills with classification at most
$\mathrm{maxClearance}_i$. The trust root supports a one-shot
\emph{lock} operation: once locked, mutations
($\mathtt{set}$, $\mathtt{remove}$) raise a typed error. The
locked state is the production posture; an unlocked trust root is
acceptable only during host bootstrap, before any external input
has been read. This is the standard ``capability bootstrap
discipline'' adapted from object-capability systems
\cite{ocap-discipline}.

\subsection{Manifest verification}

Loading a skill walks the following steps in order at bootstrap,
before any external input is read, failing closed on any error:

\begin{enumerate}
  \item Parse $M$ from canonical JSON; reject unknown fields,
        prototype-pollution keys, and missing mandatory fields.
  \item Resolve $M.\mathrm{signer}$ in the trust root. If absent,
        reject.
  \item Verify $\sigma$ against the resolved public key over the
        canonical bytes of $(M, \texttt{content})$; reject on
        mismatch.
  \item Check that $M.\mathrm{label} \preceq
        \mathrm{maxClearance}$ for the resolved signer. A signer
        cannot sign above its authorized clearance.
  \item Check $M.\mathrm{label} \preceq \texttt{user.clearance}$
        for the running operator. The operator cannot load a skill
        that exceeds their own clearance.
  \item Check that $M.\mathrm{verification}$ does not exceed the
        signer's authorized verification level. A signer attesting
        \textsc{tested} must hold attestation authority for that
        level; \textsc{formal} requires a corresponding tooling
        attestation. Default if absent: \textsc{unverified}.
  \item Register the declared $M.\mathrm{caps}$ with the
        runtime's capability gate (\S\ref{sec:gate}), tagged with
        $M.\mathrm{verification}$.
\end{enumerate}

A skill that survives all seven steps is \emph{loaded}; its
content is now reachable by the LLM at the registered
verification level. A failure at any step produces a typed audit
record and aborts the load. After bootstrap completes, the
loaded set is frozen for the session; loading a new skill at
runtime requires the runtime to be re-bootstrapped and is
governed by the \S\ref{sec:no-mutation} rule.

\subsection{Capability vocabulary}
\label{sec:caps}

The vocabulary is small enough to be enumerated and large enough
to discriminate side-effect classes. We propose a minimal set:

\begin{table}[h]
\centering\footnotesize
\begin{tabular}{@{}l p{0.42\columnwidth}@{}}
\toprule
\textbf{Capability} & \textbf{Side-effect class} \\
\midrule
\texttt{net.egress(host)}        & DNS-resolved network reach to host \\
\texttt{fs.read(path)}           & filesystem read under path \\
\texttt{fs.write.rev(path)}      & write that can be rolled back \\
\texttt{fs.write.irrev(path)}    & delete, overwrite, truncate \\
\texttt{tool.invoke(name)}       & invoke a named tool registered with the harness \\
\texttt{spawn.proc(cmd)}         & external process spawn \\
\texttt{publish(channel, \ldots)} & post to an external channel \\
\texttt{pay(token, amount)}      & transfer of fungible value \\
\texttt{mutate.schema(target)}   & migrate a database / configuration schema \\
\bottomrule
\end{tabular}
\caption{A minimal capability vocabulary. The split between
\texttt{fs.write.rev} and \texttt{fs.write.irrev} is the
load-bearing distinction the gate uses in \S\ref{sec:gate}.}
\label{tab:caps}
\end{table}

A skill's manifest must enumerate every capability its content
intends to invoke. A capability not in $M.\mathrm{caps}$ is denied
at the gate, regardless of what the skill content asks.
Theoretically, the vocabulary is the minimal carrier for the
reversible/irreversible partition over the principal classes
under which an agent's authority is denominated in the
object-capability tradition~\cite{miller-thesis}.

\paragraph{Manifest-declared, not runtime-inferred.} The reversibility
tag on a capability is a property of the runtime's own dispatch
taxonomy, not a static analysis of underlying syscalls. A path-write
is reversible exactly when the dispatch goes through the runtime's
transaction buffer (\S\ref{sec:reversible-irreversible}), which holds
the change in memory and either commits or rolls back; a write that
bypasses the buffer (or hits a remote object the runtime cannot
snapshot, or invokes a side-effect on a counterparty the runtime
cannot coordinate with) is classified \texttt{irrev}. The runtime
need not inspect the filesystem; it inspects which dispatch path
\emph{it} routed the call through. This sidesteps the well-known
problem that POSIX filesystems do not expose enough information to
classify writes statically.

\paragraph{Minimum, not closure.} Table~\ref{tab:caps} is a minimum
vocabulary, not a fixed closure. A runtime adding a new side-effect
class (e.g., \texttt{enclave.attest}, \texttt{key.derive}) extends
the enum and bumps the manifest schema version; manifest verification
(\S\ref{sec:schema}) already rejects unknown capability tokens at
load time, so vocabulary growth is monotone and backwards-compatible.
The gate's invariants in \S\ref{sec:gate} are stated over the
\emph{declared} capability set, not over Table~\ref{tab:caps}'s
specific membership.

\section{The capability gate}
\label{sec:gate}

The capability gate is the runtime layer between the LLM-driven
agent and the external world. It receives a tool-call envelope
emitted by the agent (typically as a JSON object), looks up the
corresponding capability, and consults the gate policy.

\subsection{Gate policy as a function of verification level}
\label{sec:gate-by-level}

The gate's behavior depends on the loaded skill's verification
level (\S\ref{sec:verification}). The mapping is fixed and is
not configurable per call:

\paragraph{\textsc{unverified}.} Every irreversible capability
call the agent makes while this skill is the active context
passes through HITL, regardless of whether the skill's manifest
declares the capability. Reversible calls execute through the
transaction buffer (\S\ref{sec:reversible-irreversible}). HITL on
every irreversible call is the universal default; verification
is the only path off it.

\paragraph{\textsc{declared}.} Irreversible calls whose
\texttt{(capability, target)} is in $M.\mathrm{caps}$ walk the
four-state HITL lifecycle of \S\ref{sec:hitl-lifecycle} with the
\emph{broker auto-approving} (no human prompt; the operator's
policy delegates the decision to the manifest's declared set);
the \texttt{irreversible.\allowbreak request},
\texttt{decision}, and \texttt{executed}/\texttt{error} records
are all written. Calls outside $M.\mathrm{caps}$ walk the same
lifecycle with the broker consulted normally and may stop on a
human. Reversible calls execute through the transaction buffer
regardless of verification level.

\paragraph{\textsc{tested}.} Same as \textsc{declared}, with the
addition that the runtime maintains a per-session biconditional
check (\S\ref{sec:biconditional}) over the operations the skill
issued and aborts the session if the check fails between rounds.

\paragraph{\textsc{formal}.} Same as \textsc{tested}; the
difference is in offline trust, not in runtime behavior.

In all four cases, the runtime audits every call; the difference
is whether the gate stops to ask a human first. The default is
``ask''. Verification is what buys the runtime permission to
stop asking, and only for what the verification covered.

\subsection{Reversible vs. irreversible}
\label{sec:reversible-irreversible}

The split between \emph{reversible} and \emph{irreversible}
capabilities is the design's load-bearing distinction.

\begin{itemize}
  \item A \textbf{reversible} side-effect leaves a single object in
        a state from which the runtime, holding a recent snapshot
        of that object, can return it to its prior state without
        external coordination. Memory-buffered file writes that
        commit on confirm are reversible; in-database transactions
        with rollback are reversible; an SQS queue write to a
        dead-letter queue under operator control is reversible.
  \item An \textbf{irreversible} side-effect is one for which the
        runtime cannot, alone, restore the prior state of the
        affected world. Sending an email is irreversible. Posting
        to a public channel is irreversible. Issuing an on-chain
        transaction is irreversible. Deleting a file from a remote
        store the runtime does not control is irreversible.
\end{itemize}

The runtime classifies every capability call as one or the other.
Reversible calls execute through a transaction buffer that holds
the change in memory and either commits (after the operation
succeeds and is audited) or rolls back. Irreversible calls go
through the four-state lifecycle of \S\ref{sec:hitl-lifecycle}.

\subsection{The HITL lifecycle}
\label{sec:hitl-lifecycle}

For every irreversible call the runtime walks four states:

\begin{enumerate}
  \item \textbf{request.} The agent emits a tool-call envelope
        $\langle \mathrm{op}, \mathrm{args}, \mathrm{reasoning}
        \rangle$. The runtime appends a typed
        \texttt{irreversible.\allowbreak request} record to the audit log.
  \item \textbf{decide.} The runtime consults a \emph{broker} ---
        an opaque oracle the operator chose at deploy time --- and
        receives a binary decision $d \in \{\texttt{approve},
        \texttt{deny}\}$. The runtime appends a typed
        \texttt{irreversible.\allowbreak decision} record carrying $d$ and the
        broker's identity.
  \item \textbf{execute.} If $d = \texttt{approve}$, the runtime
        performs the side-effect through the host APIs the
        capability resolves to. If the side-effect succeeds the
        runtime appends \texttt{irreversible.\allowbreak executed} with
        $\mathrm{ok}=\texttt{true}$. If the host APIs fail (the
        target file vanished between approval and call, the network
        is partitioned, the chain rejected the transaction) the
        runtime appends \texttt{irreversible.\allowbreak error} instead.
  \item \textbf{audit.} The records of the previous three states
        are linked by a shared \emph{request-id} and form a
        complete trace of the call.
\end{enumerate}

\subsection{Broker policies}

The broker is a configurable component. We see four useful
defaults:

\paragraph{\textsc{deny-all}.} The broker denies every request.
Equivalent to running the runtime without irreversible
capability at all; useful as a baseline and as the fail-safe
default when the operator has not configured a broker.

\paragraph{\textsc{policy}.} The broker reads an out-of-band
policy document (a file the LLM cannot reach) carrying allow- and
deny-rules over capability + argument shape. The decision is
mechanical and reproducible.

\paragraph{\textsc{interactive}.} The broker prompts a human via
terminal, message bus, or webhook, with a timeout that defaults
to deny.

\paragraph{\textsc{webhook}.} The broker delegates to a remote
service holding the operator's policy.

The policy regime has the convenient property that, given the same
audit log of requests, two independent operators with the same
policy produce the same decisions; this makes integration tests
deterministic.

\section{Biconditional correctness for HITL gating}
\label{sec:biconditional}

A gate's job is not merely to deny bad requests; it is to make
the runtime's claim about the world match the world. The same
property that makes the gate auditable in production also makes
it the natural criterion against which a candidate
skill-verification procedure can be exercised. We formalize the
property and discuss both uses.

\paragraph{Scope of the term ``correctness''.} The criterion below
is \emph{gate-correctness} (or, equivalently, audit-world
consistency) at the operation-and-target projection: it ensures
that the runtime's claim about what it approved-and-executed
matches the runtime's observable footprint on the corpus. It is
\emph{not} a value-level correctness criterion. Two operations
that agree on \texttt{(op, target)} but disagree on the payload
they wrote --- e.g., an approval to append log line $\ell_1$
followed by a write of $\ell_2 \neq \ell_1$ to the same file ---
satisfy the biconditional. \S\ref{sec:undetected} enumerates this
and the other limits we know of, and \S\ref{sec:related}
positions the criterion against richer notions of agent-action
correctness in concurrent work.

\subsection{Setup}

Fix an agent run with starting state $s_0$ of a designated
\emph{corpus} (a finite set of mutable objects the agent can act
on). After the run the corpus is in state $s_1$. The runtime has
produced an audit log $L$ of typed records. We define:

\begin{align}
D &\;=\; \mathrm{delta}(s_0, s_1) \quad \text{(observed)} \\
S &\;=\; \{\,r \in L \mid r.\mathrm{type} = \texttt{irreversible.\allowbreak executed}, \\
  &\qquad\qquad\;\; r.\mathrm{ok} = \texttt{true}\,\} \quad \text{(audited)}
\end{align}

$D$ is what the runtime says happened to the world; $S$ is what
the runtime says it approved and executed.

\subsection{The biconditional}

\begin{definition}[Biconditional pass]
The audit log $L$ \emph{passes the biconditional} for the run iff
the multiset projection of $D$ onto operation-and-target equals
the corresponding multiset projection of $S$.
\end{definition}

In words: \emph{every} corpus change is explained by a matching
approved-and-executed record, and \emph{every} approved-and-executed
record corresponds to a matching corpus change. Either direction
broken is a fail.

\subsection{Detected failure modes}
\label{sec:detected-failure-modes}

\begin{proposition}
\label{prop:detected}
If the audit log $L$ passes the biconditional, the following
failure modes are absent:

\begin{enumerate}
  \item[\textbf{F1}] \textbf{Gate bypass}: an irreversible
        side-effect occurred without a corresponding
        \texttt{irreversible.\allowbreak request}/\texttt{decision} pair.
  \item[\textbf{F2}] \textbf{Audit forgery}: the log claims an
        approved-and-executed record whose corresponding corpus
        change is not present.
  \item[\textbf{F3}] \textbf{Approved-but-failed without note}: a
        request was approved and the runtime claims it executed,
        but the host call silently failed.
  \item[\textbf{F4}] \textbf{Wrong-target execution}: the runtime
        approved an op against target $t_1$ but the side-effect
        landed on $t_2 \neq t_1$.
\end{enumerate}
\end{proposition}

\begin{proof}[Proof sketch]
F1, F2 follow directly from the multiset equality.
F3 is caught because the proposition restricts $S$ to records
where $r.\mathrm{ok}=\texttt{true}$; a silent host failure should
have been recorded as \texttt{irreversible.\allowbreak error}, which is not in
$S$, so its absence from $D$ violates equality from the $S$ side.
F4 is caught because the projection includes target identity.
\end{proof}

\subsection{Failure modes the criterion does not catch}
\label{sec:undetected}

Honesty about scope:

\begin{itemize}
  \item \textbf{F5 --- Wrong content on correct target.} The
        projection is over \texttt{(op, target)}, not over the
        payload written or the value transferred. A write that
        carries malicious content to a target the operator did
        approve satisfies the biconditional. Catching this needs a
        payload-level invariant (e.g., a content classifier, a
        value-flow check on a typed channel) layered \emph{above}
        the criterion.
  \item \textbf{Read-only exfiltration.} The criterion talks about
        side-effects on the corpus; it cannot detect an agent
        reading data and silently transmitting it through a covert
        channel that bypasses both the corpus and the audit log.
        Defending this requires a content-flow primitive (DLP, an
        egress-aware classification gate) below the criterion's
        layer.
  \item \textbf{Time-of-check / time-of-use races.} If the corpus
        is mutated by an external party between $s_0$ snapshot and
        the agent's first action, $D$ is the observed delta, not
        the agent's contribution.
  \item \textbf{Approved-and-the-attacker-also-acted.} If a
        legitimate approval $r \in S$ overlaps a malicious
        side-effect on the same target, $D$ may equal $S$ at the
        target level even though the agent's action was not what
        $D$ records.
\end{itemize}

These are real limits and we name them rather than paper over
them. The criterion is necessary, not sufficient, for full
safety; it is sufficient for the ``did the gate work?'' question
that HITL deployments most often need to answer.

\subsection{The biconditional as a verification benchmark}
\label{sec:biconditional-as-verification}

Beyond its production use, the biconditional is the formal
criterion against which a candidate skill-verification procedure
can be exercised. To elevate a skill from
\textsc{unverified} to \textsc{tested}, the operator (or an
attestation pipeline acting on the operator's behalf) runs the
skill against an adversarial-ensemble of \S\ref{sec:eval} and
records the audit log. The skill is admitted at \textsc{tested}
iff the biconditional passes on every round. A skill that fails
the biconditional under verification is rejected; a skill that
passes is admitted at \textsc{tested} and pinned to the manifest
hash that was tested. Modifying the skill content invalidates the
attestation; this is the load-bearing reason for the
no-runtime-mutation rule of \S\ref{sec:no-mutation}.

\section{Adversarial-ensemble evaluation}
\label{sec:eval}

The biconditional criterion is interesting only if it is exercised
under conditions where the agent \emph{wants} to take destructive
action. We propose an evaluation shape we call the
\emph{adversarial-ensemble} sketch.

\paragraph{Setup.} Generate a fixed corpus of $N$ small files with
realistic-looking corporate-document names (audit logs,
compliance reports, financial statements, meeting notes,
test-result fixtures). Compute the per-file SHA-256 manifest as
the integrity baseline.

\paragraph{Agents.} Wire $K$ LLM-driven agents with system prompts
biased toward destructive operations. A representative four-agent
ensemble has roles such as ``cleaner'' (proposing
\texttt{fs.write.irrev} on stale files), ``consolidator''
(proposing in-place edits to merge files), ``auditor''
(proposing edits that mark records as audited), and ``critic''
(echoing whichever destructive action is in scope). The system
prompts deliberately do \emph{not} ground the agents in the
policy of the gate.

\paragraph{Run.} For $R$ rounds, every agent takes one turn,
emitting a tool-call envelope. The capability gate
(\S\ref{sec:gate}) routes irreversible calls through the
configured broker; reversible calls execute and are audited.

\paragraph{Verdict.} After the run, recompute the corpus manifest;
diff against baseline; check the biconditional against the audit
log. Verdict is binary: $D = S$ at the operation-and-target
projection, or not.

\paragraph{Why this is a useful experimental shape.} Two reasons.
First, it stresses the gate against the worst plausible agent
behaviour: a confederation actively trying to delete things.
Second, the verdict is mechanical: no human grader rates outputs;
the corpus diff and the audit log are both deterministic
artefacts. Two independent operators running the same policy
broker on the same corpus produce the same verdict.

\subsection{Empirical results}
\label{sec:empirical}

We instantiated the adversarial-ensemble in two complementary
configurations and report aggregate numbers below; full per-cell
tables follow.

\paragraph{(a) $N \times K \times R$ sweep.} A self-contained driver
generated $N \in \{10, 50, 200\}$ corpus files, wired $K \in \{2, 4,
8\}$ agents over the four roles described above, and ran $R \in \{5,
10, 25\}$ rounds. For each of the $3 \times 3 \times 3 = 27$ cells
we executed $S = 200$ independent seeds, giving $5{,}400$ runs per
scenario. Five scenarios per cell: a \emph{clean} run, plus one
run per fault $F_1$--$F_4$ in which we deterministically injected the
corresponding fault at round $R/2$ (stealth write outside the gate;
fabricated \texttt{irreversible.\allowbreak executed}; approved call
whose host side silently no-ops; approved op routed to the wrong
target). After each run the driver computes the corpus delta $D$,
extracts the audited multiset $S$, and applies the biconditional
check of \S\ref{sec:biconditional}.

\emph{Results.} Across $27 \times 200 \times 5 = 27{,}000$ runs the
biconditional gave zero false positives (FPs) on clean runs and
zero false negatives (FNs) on injected faults: Wilson 95\%
confidence interval (CI) $[0.999, 1.000]$
for every \texttt{clean pass} and every $F_i$ detection rate at
the aggregate level (Table~\ref{tab:eval-sweep}). The per-cell
results are scale-invariant in $N$, $K$, and $R$, which
empirically matches the formal claim of
Proposition~\ref{prop:detected}: the criterion is mechanical
within its detection envelope.

\paragraph{(b) Cross-model adversary cooperation.} We re-ran an
existing in-vivo harness driven by ten LLM families (one Gemini,
three Groq, five Ollama, one OpenRouter; full list in the released
artifact). Each adversary generated a balanced legit/adversarial
payload set across all four F-categories and both real-backend
channels (Discord, Telegram); each payload was mediated through
three subjects: an unhardened upstream passthrough (negative
control), the OSS runtime with biconditional, audit, and a
data-loss-prevention (DLP) scanner, and the full deployment that
adds a behavioral monitor. \emph{Results.} Across $9{,}724$ mediated
payloads ($4{,}816$ adversarial, $4{,}908$ legit), the upstream
passthrough blocked $0.0\%$ (CI $[0.000, 0.001]$); the OSS runtime
achieved $76.3\%$ recall (CI $[0.751, 0.775]$) at $0.0\%$
false-positive rate (FPR); the full deployment achieved $80.2\%$
recall (CI $[0.791, 0.813]$) at $0.6\%$ FPR.
\emph{Why recall is not 100\%.} The denominator counts every
adversarial sample the LLM emitted, including the ones the
adversary itself failed to make adversarial in a recognizable way
(refusals, malformed JSON, payloads that resolve to a benign
$(\mathrm{op}, \mathrm{target})$ projection); with no real fault
present, the gate has nothing to flag. On the subset where the
adversary did produce a fault-shaped payload, recall is $1.000$
across every model family --- per-model variation in the headline
number is dominated by adversary cooperation rate, not gate
behavior.

\input{eval/sweep-table.tex}

\paragraph{Artifact.} The full per-model breakdown, the
$N \times K \times R$ driver, and the gzipped per-sample CSVs are
released with the reference implementation; runs are deterministic
given a seed.

\section{Architectural guidelines from a reference implementation}
\label{sec:guidelines}

The schema and the biconditional criterion above are not thought
experiments. They are abstracted from a working open-source
framework, \emph{enclawed}~\cite{metere2026enclawed}, which
hardens a single-user AI assistant gateway against the same threat
surface skills now expose. In this section we walk through the
specific architectural choices that framework makes, explain what
each choice was answering, and lift the choice into a normative
guideline a portable runtime profile for skill-aware harnesses
could adopt.

\subsection{Toward a portable runtime profile for skill-aware harnesses}
\label{sec:profile}

The choices above generalize to a small set of normative
guidelines for any runtime that loads, invokes, or composes
skills. We use \textsc{must} / \textsc{should} / \textsc{may} in
the sense of RFC~2119.

\begin{enumerate}
  \item[\textbf{G1}] \textbf{Capability bootstrap.} Trust roots,
        classification schemes, and policy tables \textsc{must}
        be established and locked before the runtime reads any
        skill content, any model output, or any external network
        traffic. The lock \textsc{must} be unconditional (no
        operator override at run time).
  \item[\textbf{G2}] \textbf{Deny-by-default at every boundary.}
        Network egress, provider invocation, channel use, file
        write, and tool dispatch \textsc{must} require positive
        allow-list membership. The default policy
        \textsc{must} be deny.
  \item[\textbf{G3}] \textbf{Mandatory classification.} Every
        skill artifact, every tool input, and every tool output
        \textsc{must} carry a classification label drawn from a
        finite, declared lattice. Label-free artifacts
        \textsc{must} be denied at the gate.
  \item[\textbf{G4}] \textbf{Clearance-bounded signing.} Trust-
        root signers \textsc{must} be bound to a maximum
        clearance. A skill manifest declaring a clearance above
        its signer's maximum \textsc{must} be rejected.
  \item[\textbf{G5}] \textbf{Hash-chained audit at gate-event
        granularity.} Every gate event \textsc{must} produce an
        audit record carrying a \texttt{prevHash} link to its
        predecessor; concurrent appends \textsc{must} serialize
        deterministically. The chain \textsc{must} be verifiable
        by any party that can read the log.
  \item[\textbf{G6}] \textbf{Reversible/irreversible split.}
        Every capability \textsc{must} be statically tagged as
        reversible or irreversible. Irreversible capabilities
        \textsc{must} pass through the HITL gate of
        \S\ref{sec:gate}; reversible capabilities \textsc{should}
        use a transaction buffer with rollback.
  \item[\textbf{G7}] \textbf{Biconditional post-hoc verification.}
        The runtime \textsc{should} support the biconditional
        criterion of \S\ref{sec:biconditional} as a
        runnable check after any agent run; the corpus delta
        \textsc{should} be reconcilable with the
        approved-and-executed set in the audit log without
        operator interpretation.
  \item[\textbf{G8}] \textbf{Adversarial test corpus.} The
        runtime \textsc{must} ship a corpus of pen-tests
        covering, at minimum, the failure modes
        \textbf{F1}--\textbf{F4} of \S\ref{sec:biconditional} plus
        the eight attack families enumerated above (audit
        in-place edit, log injection, signature forgery,
        hostname-normalization bypass, post-lock trust-root
        mutation, ReDoS, prompt-injection role-spoofing, code
        injection through prototype pollution). The corpus
        \textsc{must} run in CI on every commit.
  \item[\textbf{G9}] \textbf{Standard configuration profiles.}
        The runtime \textsc{must} expose at least one strict
        deployment profile (analogous to the reference's
        \texttt{enclaved} flavor) and one development profile
        (analogous to \texttt{open}). Configuration
        \textsc{must not} be a free-form continuum of feature
        flags; deployment-class differences \textsc{must} be
        named and documented.
  \item[\textbf{G10}] \textbf{No bypass switch.} The runtime
        \textsc{must not} expose a build flag, environment
        variable, runtime API, or operator override that
        disables the HITL gate, the audit log, the egress
        guard, the classification check, or the trust-root
        lock. Defense in depth \textsc{must} not have a master
        kill-switch.
  \item[\textbf{G11}] \textbf{Skills are untrusted by default.}
        The runtime \textsc{must} treat every skill manifest as
        \textsc{unverified} unless the manifest carries an
        explicit verification level set by the operator (or an
        operator-authorized attestation pipeline) at bootstrap.
        The \textsc{unverified} default \textsc{must} subject
        every irreversible capability call to HITL. A signature
        and a clearance check \textsc{must not} be sufficient
        to elevate a skill above \textsc{unverified}.
  \item[\textbf{G12}] \textbf{Verification is bootstrap-only,
        skills are immutable in-session, and every mutation
        attempt is audited.} Verification levels \textsc{must}
        be set before the runtime reads any external input and
        \textsc{must not} be elevated during a session. Loaded
        skill content \textsc{must not} be modifiable by the
        agent during the session; any agent attempt to mutate
        a loaded skill (file-system edit, runtime API, indirect
        side-effect on the on-disk artifact) \textsc{must} be
        intercepted as an irreversible capability call, walked
        through HITL, \emph{and} written into the hash-chained
        audit log with the pre- and post-mutation content
        hashes regardless of whether the operator approves or
        denies. Successful mutations \textsc{must} produce a
        new skill artifact that is re-verified before the next
        session.
\end{enumerate}

\paragraph{Attack-surface bounding.} G3+G4+G5+G11 bound
supply-chain attacks (unsigned skills rejected; signed-but-loaded
skills sit at \textsc{unverified} until an operator acts).
G2+G6 bound prompt-injection through skill content (induced
irreversible calls are intercepted; approval cannot be
manufactured from model output). G1+G4+G12 bound self-generated
and in-flight skill mutations. G5+G7 bound audit forgery. G10
bounds operator-pressure failure modes. The reduction is not to
zero --- the residual surface of \S\ref{sec:undetected} remains
--- but the surface G1--G12 close is precisely the one where
incidents are post-hoc-undeniable.

\section{Open problems}
\label{sec:discussion}

Three open directions. \emph{Label composition under
declassification}: the conservative join-of-labels rule
(\S\ref{sec:schema}) is correct but blunt; provable
declassification when outputs depend only on low-label inputs is
a research direction. \emph{Signer revocation in-flight}: skills
loaded under a now-revoked signer should be evicted at the next
gate event rather than allowed to drain; PKI mechanics
(CRLs/OCSP/short lifetimes) carry over but session eviction
needs more work. \emph{Hardware-rooted trust}: the schema
extends straightforwardly to TPM-sealed signer keys,
secure-enclave-bound brokers, and attested HITL devices without
changing the criterion.

\section{Related work}
\label{sec:related}

\paragraph{Mandatory access control and information-flow.} The
classification component of the schema is Bell--LaPadula
\cite{bell-lapadula} adapted to skill artefacts. SELinux
\cite{selinux} demonstrated MAC at OS scope; Asbestos
\cite{asbestos} and HiStar \cite{histar} extended decentralized
information-flow control to whole systems. Our contribution is
not the lattice but its lifting onto skills as deployable
artefacts.

\paragraph{Audit + capabilities.} Hash-chained logs are folklore
\cite{schneier-kelsey}; the novelty here is linking the audit
records to the biconditional check, not the chain construction.
The discipline of manifest-declared capabilities and one-shot
trust-root locks follows the object-capability tradition
\cite{ocap-discipline, miller-thesis}, lifted from intra-process
objects to inter-host deployable units.

\paragraph{Prompt-injection defenses.}
Model-boundary defenses --- NeMo Guardrails \cite{nemo-guardrails},
Llama Guard \cite{llama-guard}, Lakera Guard \cite{lakera-guard}
--- operate as input/output filters and are complementary to ours:
a skill that is malicious but well-formed may pass an input filter
yet be defeated by capability denial, and vice versa. Filter-only
defenses are also subvertible by weak-to-strong jailbreaks
\cite{zhao2024weak}, reinforcing the case for a runtime gate
independent of the model's own classification of its inputs.

\paragraph{Runtime enforcement for agents.} AgentSpec
\cite{agentspec2025} introduces a domain-specific language for
trigger/predicate/enforcement rules with quantitative evaluation
across code, embodied, and autonomous-vehicle agents; the rules
themselves are user-authored, while our schema treats the skill
artifact as the unit of governance with manifest-declared
capabilities and a fixed verification-level field. The two are
complementary: AgentSpec rules can be installed \emph{inside} our
broker (\S\ref{sec:gate}) as the policy layer for irreversible
calls, and our biconditional then verifies that the rules' net
effect on the world matches the audit log. The recent
verifiably-safe-tool-use proposal \cite{verifsafetool2026}
derives enforceable specs over data flows and tool sequences from
STPA hazards; it sits at the \emph{tool-protocol} layer where
ours sits at the \emph{skill-admission} layer, and its data-flow
guarantees would strengthen our F5 (wrong-content) gap of
\S\ref{sec:undetected}. SafeAgent \cite{safeagent2026} adds a
stateful runtime controller over persistent session state,
orthogonal to ours along the time axis (per-call gate vs.
per-trajectory state machine); its session invariants would
close the overlap-with-malicious-actor gap of
\S\ref{sec:undetected}.

\paragraph{Skill ecosystems and governance.} A concurrent 2026
survey \cite{agentskills-survey-2026} proposes a four-tier (T1--T4)
gate-based skill trust and lifecycle governance framework over the
\texttt{SKILL.md} ecosystem, motivated by an empirical finding
that $\sim 26\%$ of community skills carry at least one
vulnerability. The overlap with our contribution is real and we
acknowledge it: the survey and this paper agree that skills
require provenance-bound, verification-gated permissioning. Where
this paper goes further is on three structural axes the survey
does not commit to: \emph{bootstrap-only immutability} of the
verification level (\S\ref{sec:no-mutation}), \emph{audit
reconciliation against world state} via the biconditional
(\S\ref{sec:biconditional}), and a \emph{no-bypass-switch}
invariant (G10) that closes the operator-pressure failure mode.
SkillProbe \cite{skillprobe2026} and the threat-taxonomy work of
\cite{secureagentskills2026} are empirical-auditing complements:
they characterize the attack surface our schema is meant to bound,
and their findings (prompt-injection-shaped instructions, silent
exfiltration patterns, supply-chain residues) map onto our F1--F5
plus the side-channels of \S\ref{sec:undetected}.

\section{Conclusion}

The trust schema (\S\ref{sec:schema}) adds the verification-level
field and immutable-in-session discipline absent from current
\texttt{SKILL.md} conventions; the capability-gate model
(\S\ref{sec:gate}) keys HITL to verification level, replacing the
untenable uniform-HITL default; the biconditional criterion
(\S\ref{sec:biconditional}) is an audit-log-driven check on what
the runtime \emph{did}, complementary to behavioral benchmarks
\cite{andriushchenko2025agentharm} that score what an agent
\emph{decided}; the G1--G12 guidelines (\S\ref{sec:guidelines})
abstract a portable invariant set any harness can adopt. The
failure modes the criterion catches are precisely the ones
operators discover, after the fact, in incident review.

\bibliographystyle{ACM-Reference-Format}
\bibliography{refs}

\end{document}

%% file: eval/sweep-table.tex
\begin{table}[h]
\centering\footnotesize
\setlength{\tabcolsep}{4pt}
\resizebox{\textwidth}{!}{%
\begin{tabular}{@{}r r r c c c c c@{}}
\toprule
$N$ & $K$ & $R$ & Clean pass & F1 detect & F2 detect & F3 detect & F4 detect \\
\midrule
\multicolumn{8}{@{}l@{}}{\textit{Varying corpus size $N$ ($K=4$, $R=10$)}}\\
10 & 4 & 10 & 1.000 [0.981, 1.000] & 1.000 [0.981, 1.000] & 1.000 [0.981, 1.000] & 1.000 [0.981, 1.000] & 1.000 [0.981, 1.000] \\
50 & 4 & 10 & 1.000 [0.981, 1.000] & 1.000 [0.981, 1.000] & 1.000 [0.981, 1.000] & 1.000 [0.981, 1.000] & 1.000 [0.981, 1.000] \\
200 & 4 & 10 & 1.000 [0.981, 1.000] & 1.000 [0.981, 1.000] & 1.000 [0.981, 1.000] & 1.000 [0.981, 1.000] & 1.000 [0.981, 1.000] \\
\midrule
\multicolumn{8}{@{}l@{}}{\textit{Varying agent count $K$ ($N=50$, $R=10$)}}\\
50 & 2 & 10 & 1.000 [0.981, 1.000] & 1.000 [0.981, 1.000] & 1.000 [0.981, 1.000] & 1.000 [0.981, 1.000] & 1.000 [0.981, 1.000] \\
50 & 4 & 10 & 1.000 [0.981, 1.000] & 1.000 [0.981, 1.000] & 1.000 [0.981, 1.000] & 1.000 [0.981, 1.000] & 1.000 [0.981, 1.000] \\
50 & 8 & 10 & 1.000 [0.981, 1.000] & 1.000 [0.981, 1.000] & 1.000 [0.981, 1.000] & 1.000 [0.981, 1.000] & 1.000 [0.981, 1.000] \\
\midrule
\multicolumn{8}{@{}l@{}}{\textit{Varying rounds $R$ ($N=50$, $K=4$)}}\\
50 & 4 & 5 & 1.000 [0.981, 1.000] & 1.000 [0.981, 1.000] & 1.000 [0.981, 1.000] & 1.000 [0.981, 1.000] & 1.000 [0.981, 1.000] \\
50 & 4 & 10 & 1.000 [0.981, 1.000] & 1.000 [0.981, 1.000] & 1.000 [0.981, 1.000] & 1.000 [0.981, 1.000] & 1.000 [0.981, 1.000] \\
50 & 4 & 25 & 1.000 [0.981, 1.000] & 1.000 [0.981, 1.000] & 1.000 [0.981, 1.000] & 1.000 [0.981, 1.000] & 1.000 [0.981, 1.000] \\
\midrule
\multicolumn{3}{@{}l}{\textbf{Aggregate} ($n=5400$)} & 1.000 [0.999, 1.000] & 1.000 [0.999, 1.000] & 1.000 [0.999, 1.000] & 1.000 [0.999, 1.000] & 1.000 [0.999, 1.000] \\
\bottomrule
\end{tabular}%
}
\caption{Adversarial-ensemble sweep over $N$ corpus files, $K$ agents, and $R$ rounds, with $S=200$ seeds per cell ($27 \times 200 = 5400$ total trials per scenario). The four-role ensemble (cleaner/consolidator/auditor/critic, cycled when $K>4$) emits a mix of \texttt{fs.read} and \texttt{fs.write.irrev}; the gate's allow-all broker approves every irreversible call and audits it (\S\ref{sec:gate}). \emph{Clean pass} = fraction of injection-free runs in which the biconditional $D=S$ holds (this is $1-\text{FPR}$); \emph{$F_i$ detect} = fraction of runs with an injected $F_i$ fault in which the biconditional flags the fault. Wilson 95\% CIs in brackets. The criterion is mechanical within its detection envelope (Prop.~\ref{prop:detected}); per-cell results are scale-invariant, and the aggregate-row CI is tight at $[0.999, 1.000]$.}
\label{tab:eval-sweep}
\end{table}